\documentclass[a4paper,10pt]{article}
\usepackage{amsmath,amssymb, amsthm, graphicx, a4wide, xcolor, enumerate,amsopn}
\usepackage{arydshln}

\usepackage{tikz}
\usetikzlibrary{matrix,arrows,decorations.pathmorphing}
\usepackage{tikz-cd}

\newtheorem{theorem}{\bfseries Theorem}[section]
\theoremstyle{definition}  

\newtheorem{proposition}[theorem]{\bfseries Proposition}
\newtheorem{corollary}[theorem]{\bfseries Corollary}
\newtheorem{definition}[theorem]{\bfseries Definition}

\renewcommand{\qed}{\hspace*{\fill}\rule{1 ex}{1.5 ex}\\}
\renewenvironment{proof}{\begin{trivlist}\item[]{\bfseries Proof.}}{\qed\end{trivlist}}

\let\<=\langle
\let\>=\rangle

\def\land{\wedge}       \def\lor{\vee}

\def\oprod{\otimes}

\def\del{\triangledown}

\def\TITLE{Categorial subsystem independence as morphism co-possibility\footnote{Forthcoming in the special issue of Communications in Mathematical Physics devoted to Rudolf Haag}}
\def\AUTHOR{Zal\'an Gyenis\thanks{Budapest University of Technology and Economics, Department of Algebra and Department of Logic, E\"otv\"os Lor\'and University; gyz@renyi.hu} \and Mikl\'os R\'edei\thanks{Department of Philosophy, Logic and Scientific Method, London School of Economics and Political Science, Houghton Street, London WC2A 2AE, UK, m.redei@lse.ac.uk}}
\def\DATE{\today}
\def\ABSTRACT{This paper formulates a notion of independence of subobjects of an object in a general (i.e. not necessarily concrete) category. Subobject independence is the categorial generalization of what is known as subsystem independence in the context of algebraic relativistic quantum field theory. The content of subobject independence formulated in this paper is morphism co-possibility: two subobjects of an object will be defined to be independent if any two morphisms on the two subobjects of an object are jointly implementable by a single morphism on the larger object. The paper investigates features of subobject independence in general, and subobject independence in the category of \C algebras with respect to operations (completely positive unit preserving linear maps on $C^{\ast}$-algebras) as morphisms is suggested as a natural subsystem independence axiom to express relativistic locality of the covariant functor in the categorial approach to quantum field theory.}
\def\KEYWORDS{Algebraic relativistic quantum field theory; Category theory; Subsystem independence}

\makeatletter
\newcommand{\vek}[2][r]{%
  \gdef\@VORNE{1}
  \left[\hskip-\arraycolsep%
    \begin{array}{#1}\vekSp@lten{#2}\end{array}%
  \hskip-\arraycolsep\right]}

\def\vekSp@lten#1{\xvekSp@lten#1;vekL@stLine;}
\def\vekL@stLine{vekL@stLine}
\def\xvekSp@lten#1;{\def\temp{#1}%
  \ifx\temp\vekL@stLine
  \else
    \ifnum\@VORNE=1\gdef\@VORNE{0}
    \else\@arraycr\fi%
    #1%
    \expandafter\xvekSp@lten
  \fi}
\makeatother

\def\endef{\hfill$\square$}

\def\C{${C}^{\ast}$ }

\DeclareMathOperator{\cl}{cl}

\def\2{\mathbf{2}}

\def\bC{\mathbf{C}}
\DeclareMathOperator{\Hom}{Hom}

\def\cA{{\cal A}}
\def\cB{{\cal B}}
\def\cC{{\cal C}}

\def\cF{{\cal F}}

\def\cP{{\cal P}}

\def\cR{{\cal R}}

\def\es{\wedge}
\def\C{$C^{\ast}$-}
\def\W{$W^{\ast}$-}
\def\mfMan{\mathfrak{Man}}
\def\mfAlg{\mathfrak{Alg}}

\begin{document}

    \title{\vspace*{-1cm}\TITLE}
    \author{\AUTHOR} \date{\DATE}
    \maketitle
    \thispagestyle{empty}

    \begin{abstract}
        \ABSTRACT
        \vspace{5mm}

    \end{abstract}
    \vspace{5mm}  \normalsize


\section{Motivation}\label{sec:motiv}

Subsystem independence is a crucial notion in the specific axiomatic approach to (relativistic) quantum field theory known as ``Local Quantum Physics'' (also called ``Algebraic Quantum Field Theory''). This approach to quantum field theory was initiated by Haag and Kastler \cite{Haag-Kastler1964}, and since its inception it has developed into a rich field. (For monographic summaries see \cite{Horuzhy1990}, \cite{Haag1992}, \cite{Araki1999}; for compact, more recent reviews we refer to \cite{Buchholz2001}, \cite{Buchholz-Haag2000}, \cite{Summers2012perspective}; the papers \cite{Fredenhagen2010}, \cite{Haag2010a}, \cite{Haag2010b} recall some episodes in the history of this approach.) The key element in the approach is the implementation of locality, and ``The locality concept is abstractly encoded in a notion of independence of subsystems$\ldots$''  \cite{Brunetti-Fredenhagen2006}. It turns out that independence of subsystems of a larger system can be specified in a number of nonequivalent ways: Summers' 1990 paper \cite{Summers1990b} gives a review of the rich hierarchy of independence notions; for a non-technical review of subsystem independence concepts that include more recent developments as well see \cite{Summers2009}.

Local Quantum Physics has recently been further developed into what can properly be called `Categorial Local Quantum Physics' \cite{Brunetti-Fredenhagen-Verch2003}: in this new ``paradigm'' quantum field theory is a covariant functor from the category of certain spacetimes with isometric, smooth, causal embeddings of spacetimes as morphisms into the category of \C algebras with injective \C algebra homomorphisms as morphisms. (For a self-contained review of this approach see \cite{Fredenhagen-Reijzner2016}). This categorial approach is motivated by the desire to establish a generally covariant quantum field theory on a general, non-flat spacetime that might not have any non-trivial global symmetry. Lack of global spacetime symmetry makes it impossible to postulate covariance of quantum fields in the usual way by requiring observables to transform covariantly with respect to representations of the global symmetry group of the spacetime. Instead, general covariance is implemented in Categorial Local Quantum Physics very naturally by postulating that the functor representing quantum field theory is a \emph{covariant} functor. Locality also has to be implemented in categorial quantum field theory. This is done by formulating axioms for the covariant functor that express independence of subsystems. In the original paper formulating Categorial Local Quantum Physics   \cite{Brunetti-Fredenhagen-Verch2003} Einstein Locality (local commutativity) is taken as the expression of locality as independence. In subsequent publications \cite{Brunetti-Fredenhagen2009}, \cite{Brunetti-Fredenhagen-Paniz-Reijzner2014} a categorial version of the split property is added to the Einstein Locality axiom. It is then shown in \cite{Brunetti-Fredenhagen-Paniz-Reijzner2014} that (under the further assumption of weak additivity) the categorial split property is equivalent to the functor being extendable to a tensor functor between the tensor category formed by spacetimes with respect to disjoint union as tensor operation and the tensor category of \C algebras taken with the minimal tensor products of \C algebras.

In what sense are Einstein Locality and Einstein Locality together with the categorial split property (i.e. the tensor property) of the functor subsystem independence conditions? This is a non-trivial question, which is shown by the remark of Buchholz and Summers on Einstein Locality:
\begin{quote}
``This postulate, often called the condition of locality \cite{Haag1992}, has become one of the basic ingredients in both the construction and the analysis of relativistic theories \cite{Weinberg1995}. Yet, in spite of its central role in the theoretical framework, the question of whether locality can be deduced from physically meaningful properties of the physical states has been open for more than four decades.'' \cite{Buchholz-Summers2005}
\end{quote}
In particular, Buchholz and Summers point out \cite{Buchholz-Summers2005} that some of the standard notions of subsystem independence (such as \C independence) do \emph{not} entail Einstein Locality. And conversely: Einstein Locality alone does not entail \C independence. Nor does Einstein Locality, in and by itself, entail the independence condition known as prohibition of superluminal signaling: local commutativity of \C algebras pertaining to spacelike separated spacetime regions only entails prohibition of superluminal signaling with respect to measurements of local observables representable by the projection postulate and by the operations given by local Kraus operators -- but not with respect to general operations that do not have a local Kraus representation (this was shown in \cite{Redei-Valente2010}). Einstein Locality alone also does not ensure another subsystem independence called operational \C independence: That any two (non-selective) operations (completely positive, unit preserving linear maps on \C algebras) performed on spacelike separated subsystems of a larger system are jointly implementable as a single operation on the larger system \cite{Redei-Summers2009}, \cite{Redei2010FoundPhys}. Since the tensor product of two operations is again an operation (\cite{Blackadar2005}[p. 190], see also Proposition 9. in \cite{Redei-Summers2009}), the tensorial property of the functor does entail operational \C independence of the components of the tensor product within the tensor product algebra; however, the tensor property entails more than this: it entails operational \C independence \emph{in the product sense} \cite{Redei-Summers2009}, which is a strictly stronger condition than simple operational \C independence. But for the purposes of expressing locality in quantum field theory subsystem independence in the sense of operational \C independence does not have to be implemented in the strong form of requiring existence of a joint \emph{product} extension of operations on subsystems.

Thus requiring only Einstein Locality of the functor seems too weak, imposing the categorial split property (i.e. demanding the functor to be tensorial) seems to demand a bit more than needed to implement locality interpreted as subsystem independence in Categorial Local Quantum Physics. What is then the right concept of subsystem independence that expresses locality in Categorial Local Quantum Physics?

The subsystem independence hierarchy in Local Quantum Physics suggests a general concept of subsystem independence that has a natural formulation in terms of categories objects of which are sets with morphisms as maps: independence as morphism co-possibility. According to this independence concept two objects are independent in a larger object with respect to a class of morphisms if any two morphisms on the two smaller objects have a joint extension to a morphism on the larger object. This independence notion appeared in \cite{Redei2014SHPMP} and was suggested in \cite{Redei2016CatLocNagoya} as a possible axiom to require in Categorial Local Quantum Physics. Taking specific subclasses of the operations  as the class of morphisms, one can recover the standard concepts of subsystem independence in the independence hierarchy as special cases of morphism co-possibility (see \cite{Redei2014SHPMP}).

The way subsystem independence as morphism co-possibility was formulated above and in the papers \cite{Redei2014SHPMP} and \cite{Redei2016CatLocNagoya} is not entirely satisfactory however because it is not purely categorial: in a general category objects are not necessarily sets and morphisms are not necessarily functions -- not every category is a concrete category (e.g. the real numbers $\mathbb{R}$ regarded as a poset category) \cite{Awodey2010}. In a general category subsystem independence as morphism co-possibility should be formulated as \emph{subobject} independence with respect to some class of morphisms. The aim of the present paper is to define subobject independence in this way, as morphism co-possibility in a general category, and to investigate the basic properties of such an independence notion. This notion is of interest in its own right and, after defining it in section \ref{sec:subobj-indep}, we give several examples of this sort of independence in different categories in section \ref{sec:examples}. Section \ref{sec:tensor} proves some propositions on the relation of subobject independence and tensor structure in a category. In section \ref{sec:opind} subobject independence is specified in the context of the category of \C algebras taken with the class of operations between \C algebras. The resulting notion of operational independence is suggested then in section \ref{sec:functor-of-QFT} as a possible axiom to express relativistic locality of the covariant functor describing a generally covariant quantum field theory.

\section{Categorial independence of subobjects\label{sec:subobj-indep}}

\def\Ob{\mathsf{Ob}}
\def\Mor{\mathsf{Mor}}
\def\MorH{\mathsf{H}}
\def\Iso{\mathsf{Iso}}
\def\Aut{\mathsf{Aut}}

\def\Hom{\mathsf{Hom}}

In this section $\bC = (\Ob, \Mor)$ denotes a general category, and $\Hom$ is a subclass of $\Mor$ such that $(\Ob, \Hom)$ also is a category.
It is not assumed that $\bC$ is a concrete category; i.e. that it is categorically equivalent to a category objects of which are sets and monomorphisms are functions. Morphisms in $\Hom$ will be referred to as $\Hom$-morphisms, morphisms in $\Mor$ will be called $\Mor$-morphisms. We wish to define a notion of independence of subobjects $A,B$ of an object $C$, where the concept of subobject is understood with respect to $\Hom$-morphisms, and the independence expresses that any two $\Mor$-morphisms on the $\Hom$-subobjects $A$ and $B$ are jointly implementable by a single $\Mor$-morphism on $C$. The two morphism classes $\Hom$ and $\Mor$ should be considered as variables in this categorial concept of independence: Choosing different morphism classes one obtains independence notions contents of which can vary considerably.

Recall that a $\Mor$-morphism $f: A\to B$ is a monomorphism (``mono", for short) if for any object $C\in\Ob$ and any $\Mor$-morphisms
$g_1,g_2: C\to A$ it holds that $g_1f=g_2f$ implies $g_1=g_2$. Monomorphisms are the categorial equivalents of injective functions.

The notion of $\Hom$-subobject is formulated in terms of $\Hom$-monomorphisms: a $\Hom$-subobject of an object $X$ is an equivalence class of
$\Hom$-monomorphisms $\Hom\ni i_A:A\to X$ where $i_A$ is defined to be equivalent to $i_B:B\to X$ if there is an $\Hom$-isomorphism $h:A\to B$ such that $hi_B = i_A$ and $h^{-1}i_A=i_B$. In what follows, $|i_A|_{\Hom}$ denotes the equivalence class of
$\Hom$-morphisms equivalent to $i_A$.

\begin{definition}[$\Mor$-independence of $\Mor$-morphisms]\label{def:independence-morphisms}
	$\Mor$-morphisms
	\begin{tikzcd}A \arrow[->]{r}{f_A} & X \arrow[<-]{r}{f_B} & B \end{tikzcd}
	are called $\Mor$-independent if for any two $\Mor$-morphisms
	\begin{tikzcd}A \arrow{r}{\alpha_A} & A\end{tikzcd},
	\begin{tikzcd}B \arrow{r}{\alpha_B} & B\end{tikzcd} there is $\Mor$-morphism
	\begin{tikzcd}X \arrow{r}{\alpha} & X\end{tikzcd}
	such that the diagram below commutes.
	\begin{center}
		\begin{tikzcd}[column sep=large,row sep=large]
			A \arrow[->]{r}{f_A} \arrow{d}{\alpha_A} & X\arrow[dotted]{d}{\alpha}\arrow[<-]{r}{f_B} & B\arrow{d}{\alpha_B} \\
			A \arrow[->]{r}{f_A} & X\arrow[<-]{r}{f_B} & B
		\end{tikzcd}
	\end{center}
	\endef	
\end{definition}

We are now in the position to give the definition of the central concept of this paper:

\begin{definition}[$\Mor$-independence of $\Hom$-subobjects]\label{def:subobj-indep}
	Two $\Hom$-subobjects of object $X$ represented by the two equivalence classes
	$|f_1|_{\Hom}$ and $|f_2|_{\Hom}$ are called $\Mor$-independent if
	any two $\Hom$-monomorphisms
	\[\begin{tikzcd}D_1 \arrow[>->]{r}{g_1} & X\arrow[<-<]{r}{g_2} & D_2 \end{tikzcd} \]
	$g_1\in |f_1|_{\Hom}$ and $g_2\in |f_2|_{\Hom}$ are $\Mor$-independent.
	\endef
\end{definition}

The content of $\Mor$-independence of $\Hom$-subobjects is that two $\Hom$-subobjects of object $C$ are $\Mor$-independent if and only if \emph{any two} $\Mor$-morphisms on \emph{any} representations of the $\Hom$-subobjects are jointly implementable by a \emph{single} $\Mor$-morphism on $C$. This independence concept expresses the independence of the of the substructures that are invariant with respect to $\Hom$ from the perspective of the structural properties embodied in the $\Mor$-morphisms. (See examples in section \ref{sec:examples}.)

The next proposition is useful when it comes to determine whether two subobjects are independent. By definition, independence of $\Hom$-subobjects implies $\Mor$-independence of any of their representatives. The following proposition states
the converse: if one pair of representatives of two $\Hom$-subobjects are
$\Mor$-independent, then the two $\Hom$-subobjects are $\Mor$-independent.

\begin{proposition}\label{prop:one-morphism-enough} 
	If for $\Hom$-monomorphisms
	\begin{tikzcd}C_1 \arrow[>->]{r}{f_1} & X \arrow[<-<]{r}{f_2} & C_2 \end{tikzcd}
	and
	\begin{tikzcd}D_1 \arrow[>->]{r}{g_1} & X \arrow[<-<]{r}{g_2} & D_2 \end{tikzcd}
	we have $|f_1|_{\Hom}=|g_1|_{\Hom}$ and $|f_2|_{\Hom}=|g_2|_{\Hom}$
	(i.e. the $\Hom$-monomorpisms $f_i$ and $g_i$ ($i=1,2$) represent the
	same $\Hom$-subobject),
	then $f_1$ and $f_2$ are $\Mor$-independent if and only if $g_1$ and $g_2$
	are $\Mor$-independent.
\end{proposition}
\begin{proof}
	Suppose
	\begin{tikzcd}C_1 \arrow[>->]{r}{f_1} & X\arrow[<-<]{r}{f_2} & C_2 \end{tikzcd}
	are $\Mor$-independent and consider the diagram below.
	\begin{center}
	\begin{tikzpicture}
	  \matrix (m) [matrix of math nodes,row sep=2em,column sep=3em,minimum width=2em]
	  {
		 D_1  &     &   &     & D_2 \\
		      & C_1 & X & C_2 &     \\
		      & C_1 & X & C_2 &     \\
		 D_1  &     &   &     & D_2 \\};
	  \path[-stealth]
		(m-1-1) edge node [above] {$i_1$} (m-2-2)
				edge[>->,bend left] node [above] {$g_1$} (m-2-3)
				edge node [left] {$\alpha_1$} (m-4-1)
		(m-2-2) edge[>->] node [above] {$f_1$} (m-2-3)
				edge [dashed] node [left] {$\beta_1$} (m-3-2)
		(m-3-2) edge node [above] {$j_1$} (m-4-1)
				edge[>->] node [above] {$f_1$} (m-3-3)
		(m-4-1) edge[>->,bend right] node [below] {$g_1$} (m-3-3)
		(m-2-3) edge [dashed] node [left] {$\gamma$} (m-3-3)
		(m-2-4) edge[>->] node [above] {$f_2$} (m-2-3)
				edge [dashed] node [right] {$\beta_2$} (m-3-4)
		(m-3-4) edge[>->] node [above] {$f_2$} (m-3-3)
				edge node [right] {$j_2$} (m-4-5)
		(m-1-5) edge[>->,bend right] node [above] {$g_2$} (m-2-3)
				edge node [above] {$i_2$} (m-2-4)
				edge node [right] {$\alpha_2$} (m-4-5)
		(m-4-5) edge[>->,bend left] node [above] {$g_2$} (m-3-3);
	\end{tikzpicture}	
	\end{center}
	Since $|f_1|_{\Hom}=|g_1|_{\Hom}$ and $|f_2|_{\Hom}=|g_2|_{\Hom}$
	there are $\Hom$-isomorphisms $i_1, j_1$ and $i_2, j_2$
	as figured. Take and arbitrary $\Mor$-morphism
	$\alpha_1:D_1\to D_1$.
	Let
\begin{eqnarray}
\beta_1 &=& i_1^{-1}\alpha_1j_1^{-1}:C_1\to C_1\\
\beta_2 &=& i_2^{-1}\alpha_2 j_2^{-1}: C_2\to C_2
\end{eqnarray}
By assumption
	\begin{tikzcd}C_1 \arrow[>->]{r}{f_1} & X\arrow[<-<]{r}{f_2} & C_2 \end{tikzcd}
	are $\Mor$-independent, therefore there is a suitable $\Mor$-morphism $\gamma:X\to X$. Then we obtain
	\[ g_1\gamma=i_1f_1\gamma = i_1\beta_1f_1 = i_1\beta_1j_1g_1 = \alpha_1g_1 \]
	and similarly
	\[ g_2\gamma=i_2f_2\gamma = i_2\beta_2f_2 = i_2\beta_2j_2g_2 = \alpha_2g_2 \]
	This completes the proof.
\end{proof}


Our next proposition formulates a very natural necessary condition for independence. The content of the necessary condition can be illustrated
on the example of the category $\bC$ of structures. Let $A$ and $B$ substructures of $C$. If we take two morphisms $\alpha_A:A\to A$, $\alpha_B:B\to B$, then a morphism $\gamma: C\to C$ that extends {\em both} $\alpha_A$ \emph{and} $\alpha_B$ can exist
only in the case when $\alpha_A$ and $\alpha_B$ act on $Y=A\cap B$ exactly the same way, i.e. if one has
\begin{equation}\label{eq:cond-identity-on-cap}
\alpha_A\upharpoonright Y \;=\; \alpha_B\upharpoonright Y
\end{equation}
The next proposition we wish to establish expresses this condition in the case of every category. To state the proposition, first we formulate the condition (\ref{eq:cond-identity-on-cap}) in general categorial terms. Since the intersection $Y$ in the category of structures is the pullback $A\times_{C}B$, for the next definition it is assumed that pullbacks exist in $\bC$.

\begin{definition}[$\Mor$-compatibility]
	We say that $\Mor$-morphisms
	\begin{tikzcd}A \arrow[->]{r}{f_A} & C \arrow[<-]{r}{f_B} & B \end{tikzcd}
	are $\Mor$-com\-pa\-tib\-le if the diagram
	\begin{center}
	\begin{tikzpicture}
		  \matrix (m) [matrix of math nodes,row sep=4em,column sep=4em,minimum width=2em]
		  {
			   & A\times_C B  & \\
			 A  & C & B\\
			 A  & C & B\\};
		  \path[-stealth]
			(m-1-2) edge[->] node [left]  {$p_{A}$} (m-2-1)
					edge[->] node [right] {$p_{B}$} (m-2-3)
			(m-2-1) edge[->] node [above] {$f_{A}$} (m-2-2)
			(m-2-3) edge[->] node [above] {$f_{B}$} (m-2-2)
			(m-2-1) edge node [left] {$\alpha_A$} (m-3-1)
			(m-2-3) edge node [left] {$\alpha_B$} (m-3-3)
			(m-3-1) edge[->] node [below] {$f_{A}$} (m-3-2)
			(m-3-3) edge[->] node [below] {$f_{B}$} (m-3-2);
	\end{tikzpicture}\end{center}
	commutes for all $\Mor$-morphisms $\alpha_A, \alpha_B$.
	Here $A\times_C B$ is the pullback.
	\endef
\end{definition}
The next proposition states the sought-after \emph{necessary} condition for independence in a general category:

\begin{proposition}
	If \begin{tikzcd}A \arrow[->]{r}{f_A} & C \arrow[<-]{r}{f_B} & B \end{tikzcd}
	are $\Mor$-independent, then they are also $\Mor$-compatible.
\end{proposition}
\begin{proof}
	Consider the diagram, where $\alpha_A, \alpha_B$ are arbitrary $\Mor$-morphisms and
	$A\times_C B$ is the pullback.
	\begin{center}\begin{tikzpicture}
		  \matrix (m) [matrix of math nodes,row sep=3em,column sep=4em,minimum width=2em]
		  {
			    & A\times_C B & \\
			 A  & C & B\\
			 A  & C & B\\};
		  \path[-stealth]
			(m-1-2) edge[->] node [left] {$i_A$} (m-2-1)
					edge[->] node [left] {$i_B$} (m-2-3)
			(m-2-1) edge[->] node [above] {$f_{A}$} (m-2-2)
			(m-2-3) edge[->] node [above] {$f_{B}$} (m-2-2)
			(m-2-1) edge node [left] {$\alpha_A$} (m-3-1)
			(m-2-3) edge node [left] {$\alpha_B$} (m-3-3)
			(m-3-1) edge[->] node [below] {$f_{A}$} (m-3-2)
			(m-3-3) edge[->] node [below] {$f_{B}$} (m-3-2)
			(m-2-2) edge [dashed] node [right] {$\gamma$} (m-3-2);
	\end{tikzpicture}\end{center}
	We need to show that the diagram without the dashed arrow commutes.
	By $\Mor$-independence, for $\Mor$-morphisms $\alpha_A, \alpha_B$ there exists a
	suitable $\Mor$-morphism $\gamma$. Then
	\[ i_A\alpha_Af_A = i_Af_A\gamma = i_Bf_B\gamma = i_B\alpha_Bf_B \]
	which we had to show.
\end{proof}

For completeness we note that the existence of the pullback $A\times_C B$ in the definition of $\Mor$-compatibility
could be relaxed by replacing the pullback $A\times_C B$ in the diagram by any $Y$ that can be mapped
into $A$ and $B$ (i.e. there are $\Mor$-arrows $Y\to A$ and $Y\to B$) and universally quantifying over $Y$.

In the next proposition let $\oplus$ be a coproduct in the category $(\Ob,\Mor)$, i.e.
$X_1\oplus X_2$ be an element such that there exist $\Mor$-morphisms
(called the coproduct injections)
\[ i_1: X_1\to X_1\oplus X_2\quad \text{ and }\quad i_2: X_2\to X_1\oplus X_2\] having the universal property.

\begin{proposition}\label{prop:coprodindep}
	Coproduct injections \begin{tikzcd}X_1 \arrow[->]{r}{i_1} & X_1\oplus X_2 \arrow[<-]{r}{i_2} & X_2 \end{tikzcd} are $\Mor$-independent.
\end{proposition}
\begin{proof}
	Let \begin{tikzcd}X_1 \arrow[->]{r}{i_1} & X_1\oplus X_2 \arrow[<-]{r}{i_2} & X_2 \end{tikzcd} be a coproduct with coproduct injections $i_1$ and $i_2$.
	From the diagram below on the left-hand side, by composing arrows, one gets the
	diagram on the right-hand side which is a coproduct diagram, therefore
	a suitable $m$ with the dotted arrow (which is the copair $[m_1i_1,m_2i_2]$)
	exists and completes the proof.
	\begin{center}
		\begin{tabular}{cc}
			\begin{tikzpicture}
			  \matrix (m) [matrix of math nodes,row sep=3em,column sep=4em,minimum width=2em]
			  {
				 X_1  & X_1 \oplus X_2  & X_2\\
				 X_1  & X_1 \oplus X_2  & X_2\\};
			  \path[-stealth]
				(m-1-1) edge[->] node [above] {$i_{1}$} (m-1-2)
						edge node [left]  {$m_1$} (m-2-1)
				(m-2-1) edge[->] node [below] {$i_{1}$} (m-2-2)
				(m-1-3) edge[->] node [above] {$i_{2}$} (m-1-2)
						edge node [right] {$m_2$} (m-2-3)
				(m-2-3) edge[->] node [below] {$i_{2}$} (m-2-2);
			\end{tikzpicture}	
			& \qquad\qquad
			\begin{tikzpicture}
			  \matrix (m) [matrix of math nodes,row sep=3em,column sep=4em,minimum width=2em]
			  {
				 X_1  & X_1 \oplus X_2  & X_2\\
				      & X_1 \oplus X_2  &    \\};
			  \path[-stealth]
				(m-1-1) edge[->] node [above] {$i_{1}$} (m-1-2)
						edge node [below] {$m_1i_{1}$} (m-2-2)
				(m-1-2) edge [dotted] node [right] {$m$} (m-2-2)
				(m-1-3) edge[->] node [above] {$i_{2}$} (m-1-2)
						edge node [below] {$m_2i_{2}$} (m-2-2);
			\end{tikzpicture}				
		\end{tabular}
	\end{center}
\end{proof}

We remark that coproduct injections in general are not necessarily monic,
however, in certain categories (such as extensive or distributive categories) coproduct injections are automatically monic.


\section{Examples of subobject independence\label{sec:examples}}

\subsection{Sets\label{sec:examples-sets}}
$\mathbf{Set}$ is the category of sets as objects with functions as $\Mor$-morphisms. Let $\Hom = \Mor$ and consider \begin{tikzcd}A \arrow[>->]{r}{f_A} & C \arrow[<-<]{r}{f_B} & B \end{tikzcd}.
Speaking about subobjects we may assume $A,B\subseteq C$, that is, $f_A$ and $f_B$ are the inclusion mappings.
The pullback $A\times_C B$ is just the intersection $A\cap B$. $A$ and $B$ are $\Hom$-compatible if and only if
$A\cap B=\emptyset$ since otherwise one could take permutations of $A$ and $B$ that act differently on the intersection.
It is straightforward to check that $A$ and $B$ are $\Mor$-independent if and only if they are disjoint.

\subsection{Vector spaces}

Let $\mathbf{Vect}_{\mathbb{F}}$ be the category of vector spaces over the field $\mathbb{F}$
with linear mappings as $\Mor$-morphisms. Take $\Hom = \Mor$. If $C$ is a
vector space and $A$, $B$ are subspaces then the pullback $A\times_C B$ is the subspace $A\cap B$. Recall that two
subspaces $A$, $B$ are linearly independent if and only if $A\cap B = \{0\}$. We claim that $\Mor$-independence and linear independence
coincide. Take two $\Mor$-morphisms $\alpha_A:A\to A$ and $\alpha_B:B\to B$.
Then $\alpha_A$ and $\alpha_B$ act on the bases $\<a_i:i\in I\>=A$ and
$\<b_j:j\in J\>=B$. Any function defined on bases can be extended to a
linear mapping,
therefore $\alpha_A$ and $\alpha_B$ have a common extension $\gamma:\<A\cup B\>\to\<A\cup B\>$ if and only if they act on
$A\cap B$ the same way. As $\alpha_A$, $\alpha_B$ were arbitrary, the latter condition is equivalent to $A\cap B=\{0\}$.
Finally, one can extend the set $\{a_i, b_j:i\in I, j\in J\}$ to a basis of $C$ and extend $\gamma$ to be defined on the entire $C$.
A moment of thought shows that $\Mor$-compatibility is also equivalent to $A\cap B=\{0\}$.

\subsection{Pregeometries (Matroids)}

Pregeometries (or matroids in the combinatorial terminology) are defined in order
to capture the notion of independence in a very general framework. Formally, a pregeometry is a tuple
$(X, \cl)$ where $X$ is a set and $\cl:\wp(X)\to\wp(X)$ is a closure operator having a finite character satisfying the
Steinitz exchange principle. Independence and basis can be defined as in vector spaces. A morphism between
two pregeometries $f:(X, \cl_X)\to (Y,\cl_Y)$ is a function $f:X\to Y$ that preserves closed sets, that is,
it satisfies $\cl_X( f^{-1}[Z] ) = f^{-1}[\cl_Y(Z)]$ for all $Z\subseteq Y$. If the two closure operators are topological closure,
then morphisms are just the continuous functions. If $(X, \cl)$ is a pregeometry, then a sub-pregeometry is a tuple
$(Y,\cl\upharpoonright Y)$ where $Y\subseteq X$ is closed: $Y=\cl(Y)$. We denote sub-pregeometries by $Y\leq X$.
Let $\mathsf{Pregeom}$ be the category of pregeometries with $\Mor = \Hom$ as described above. Then $\Mor$-independence of $A,B\leq C$ coincides with independence of $A$ and $B$ in the pregeometry
sense. The proof is similar to that of the vector space case. $A, B\leq C$ are independent if and only if $A\cap B = \cl(\{\emptyset\})$.
This holds only if neither $A$ nor $B$ has basis elements in $A\cap B$. In this case any basis of $A$ and $B$ can be concatenated
and extended to a basis of $C$ in the same way as in the case of vector spaces. The result follows then from the observation
that any $\Mor$-morphism can be identified with an action on the elements of a basis.

\subsection{Boolean algebras\label{eg:Boolean}} Let $\mathbf{Bool}$ be the category of Boolean algebras as objects with {\em injective} homomorphisms as
$\Mor$-morphisms. As before, we set $\Hom=\Mor$. Two subalgebras $A,B\leq C$ are called Boole-independent
if for all $a\in A$, $b\in B$ we have $a\land b\neq 0$ provided $a\neq 0\neq b$.  Boole independence is logical independence if the Boolean algebras are viewed as the Tarski-Lindenbaum algebra of a classical propositional logic: $a\land b\neq 0$ entails that there is an interpretation on $C$ that makes $a\land b$ hence both $a$ and $b$ true; i.e. any two propositions that are not contradictions can be jointly true in some interpretation. How is this Boolean (logical) independence related to $\Mor$-independence? The connection between Boole-independence and $\Mor$-independence is a bit more subtle than in the previous examples.

\begin{itemize}
\item[(1)] $\Mor$-independence does not imply Boole-independence. Consider the case when $C$ is finite, $\{c_1$, $\ldots$, $c_n\}$
is the set of atoms of $C$ and the subalgebras $A$ and $B$ are generated by distinct set of atoms $A = \<c_1,\ldots, c_k\>$,
$B = \<c_{k+1}, \ldots, c_n\>$. Clearly $A$ and $B$ are not Boole-independent. However, any $\Mor$-morphisms (i.e. automorphism, because in the finite case every injective homomorphism into itself is an automorphism) of $A$ (resp. $B$)
comes from a permutation of atoms generating $A$. Conversely any permutation of atoms extend to an automorphism.
Given automorphisms $\alpha_A$ and $\alpha_B$ of $A$ and $B$, respectively, give rise to a permutation of all the atoms of $C$
which extends to an automorphism of $C$. Consequently, $A$ and $B$ are $\Mor$-independent.

\item[(2)] Boole-independence of $A,B\leq C$ implies $\Mor$-independence if $A\cup B$ generates $C$.
A Boolean algebra $C$ is the internal sum of the subalgebras $A$ and $B$ just in case the union $A\cup B$ generates $C$
and whenever $a\in A$, $b\in B$ are non-zero elements, then $a\land b\neq 0$.
(Internal) sum of Boolean algebras is just the coproduct of the algebras (up to isomorphism) whence by
Proposition \ref{prop:coprodindep} $\Mor$-independence of $A$ and $B$ follows.

If $A, B\leq C$ are Boole-independent but $A\cup B$ does not generate $C$ (i.e. the internal sum $A\oplus B$
of $A$ and $B$ is a proper subalgebra of $C$), then a similar argument shows that
any $\Mor$-morphism $\alpha_A: A\to A$ and $\alpha_B: B\to B$ can be {\em jointly} extended to a $\Mor$-morphism $\gamma: A\oplus B\to A\oplus B$. The question whether
$\gamma$ can be further extended to an $\Mor$-morphism $C\to C$ is non-trivial and is related to
the injectivity of $C$. Injective Boolean algebras are essentially the complete ones in the category $\mathbf{Bool}$
\cite{Pierce2014}[p.117, 16(c,e)].
Consequently Boole-independence implies $\Mor$-independence in any complete Boolean algebra.
\end{itemize}

\subsection{Logical and categorial subobject independence in quantum logic}
\def\OML{\mathbf{OML}}

Logical independence is meaningful in categories of lattices that are not distributive. The relevant examples for physics are the von Neumann lattices (in particular Hilbert lattices) that are interpreted as quantum logic: If $\cA$ and $\cB$ are two von Neumann subalgebras of von Neumann algebra $\cC$, and $\cP(\cA), \cP(\cB)$ and $\cP(\cC)$ denote the corresponding orthomodular lattices of projections, then $\cP(\cA)$ and $\cP(\cB)$ can be defined to be logically independent if $a\es b\not=0$ whenever $\cP(\cA)\ni a\not=0$ and $\cP(\cB)\ni b\not=0$ (\cite{Redei1995a}, \cite{Redei1995b} \cite{Redei1998}[Section 11]). Taking the category of von Neumann lattices with orthomodular lattice homomorphisms as morphisms the notion of subobject independence becomes meaningful and the problem of relation of logical independence of von Neumann lattices and the subobject independence emerges in this category just like in the category of Boolean algebras. We clarify here the relation of logical independence to subobject independence in the context of general orthomodular lattices \cite{Kalmbach1983}.

Let $\OML$ be the category of orthomodular lattices as objects with injective ortho-homomorphisms as $\Mor$-morphisms. Take $\Hom=\Mor$.
Logical independence of orthomodular sublattices $A$ and $B$ of the orthomodular lattice $C$ is defined as in case of Boolean algebras: $a\es b\neq0$ whenever $A\ni a\not=0$ and $B\ni b\not=0$. The connection between logical independence and $\Mor$-independence in this general context is similar to the one in the category of Boolean algebras. To describe the relation, recall first the notion of internal direct sum for lattices (see e.g. \cite{Maeda1949}): If $L$ is a lattice (not necessarily orthomodular) and $x,y\in L$, then write $x\del y$ if for all $z\in L$ we have $(x\lor z)\land y = z\land y$. Clearly $x\land y = 0$ implies $x\del y$.
Let $S$ and $Q$ be {\em subsets} of $L$. We say that $L$ is the internal direct sum of $S$ and $Q$ (and we write $L = S\oplus Q$) if
\begin{enumerate}
	\item[(1)] each $x\in L$ can be written as $x=s\lor q$ with $s\in S$ and $q\in Q$;
	\item[(2)] $s\in S$, $q\in Q$ entails $s\del q$.
\end{enumerate}
If $S$ and $Q$ are (orthomodular) lattices, then their direct product is an (orthomodular) lattice, and there is a natural (ortho)-isomorphism between their direct product and their internal direct sum given by $(s,q)\mapsto s\lor q$ (see \cite{Maeda1949}). It follows that any homomorphisms given on the direct summands $S$ and $Q$ extends to a homomorphism on their internal direct sum.
We have then the following characterization of the relation of logical independence and subobject independence in the category of orthomodular lattices:
\begin{itemize}
\item[(1)] $\Mor$-independence does not imply logical independence of sub-orthomodular lattices. This follows from what was said about Boolean algebras in section \ref{eg:Boolean} because every Boolean algebra is an orthomodular lattice and we saw that $\Mor$-independence does not imply Boole-independence. ($\mathbf{Bool}$ is a complete subcategory of $\OML$).

\item[(2)] Logical independence of $A,B\leq C$ implies $\Mor$-independence if $A\cup B$ generates $C$. In this case $C = A\oplus B$ is the internal sum of $A$ and $B$ since logical independence ensures $a\del b$ for each $a\in A$, $b\in B$. On the other hand, each $x\in C$ can be written as $x=a\lor b$ with $a\in A$ and $b\in B$ as $A\cup B$ generates $C$.

If $A, B\leq C$ are logically independent but $A\cup B$ does not generate $C$ (i.e. the internal sum $A\oplus B$ of $A$ and $B$ is a proper subalgebra of $C$), then a similar argument shows that any $\Mor$-morphism $\alpha_A: A\to A$ and $\alpha_B: B\to B$ can be {\em jointly} extended to a $\Mor$-morphism $\gamma: A\oplus B\to A\oplus B$. The question whether
$\gamma$ can be further extended to an $\Mor$-morphism $C\to C$ is non-trivial and is related to the injectivity of $C$. We are not aware of any useful characterization of injective objects in $\OML$.
\end{itemize}



\section{Categorial subobject independence and tensor product structure\label{sec:tensor}}

Components of tensor products are typically regarded ``independent" within the tensor product. The paradigm example is the standard product of probability measure spaces with the product measure on the product of the component measurable spaces. In this section we investigate the relation of categorial subobject independence and the tensor product structure in a category. We will see that categorial subobject independence of the components of the tensor product is not automatic. We will however isolate conditions on the tensor category that entail subobject independence of the components in the tensor product (Proposition \ref{cor:tensor-sufficient}).

Recall first the definition of a tensor product in a category (cf. Section 7.8 in \cite{Awodey2010})


\begin{definition}
	A bifunctor $\oprod:\bC\times\bC\to\bC$ is a tensor product if it is
	associative up to a natural isomorphism and there is an element $I$
	that acts as a left and right identity (up to isomorphism).\endef
\end{definition}

A category with a tensor product $(\bC, \oprod)$ is a tensorial category (monoidal category) if
$\oprod$ satisfies the pentagon and triangle axioms.
If a category has products or coproducts for all finite sets of objects, then the category
can be turned into a tensor category by adding the product or coproduct as a bifunctor (due to
the universal property of products and coproducts).

For the next definition suppose that $(\bC, \oprod)$ is such that for any two objects $A$, $B$ there
are $\Mor$-morphisms \begin{center}\begin{tikzcd} A \arrow[->]{r}{i_A} & A\oprod B \arrow[<-]{r}{i_B} & B \end{tikzcd}\end{center}

\begin{definition}[$\oprod$-independence of $\Mor$-morphisms]\label{def:tensindep}
	$\Mor$-morphisms \begin{tikzcd}A \arrow[->]{r}{f_A} & C \arrow[<-]{r}{f_B} & B \end{tikzcd}
	are called $\oprod$-independent if there exists a $\Mor$-morphism
	$h:A\otimes B\to C$ such that the following diagram commutes.
	\begin{center}
		\begin{tikzpicture}
		  \matrix (m) [matrix of math nodes,row sep=3em,column sep=4em,minimum width=2em]
		  {
		     A & A\otimes B & B  \\
		       & C &  \\};
		  \path[-stealth]
		    (m-1-1) edge node [left] {$f_A$} (m-2-2)
		    (m-1-3) edge node [right] {$f_B$} (m-2-2)
		    (m-1-1) edge[->] node [above] {$i_A$} (m-1-2)
		    (m-1-3) edge[->] node [above] {$i_B$} (m-1-2)
		    (m-1-2) edge [dashed] node [right] {$h$} (m-2-2);
		\end{tikzpicture}	
	\end{center}\endef
\end{definition}

If $\oprod$ is the coproduct, then the universal property of coproducts implies the existence of such a $h$ in the definition.

$\Mor$-independence of components of tensor products is not automatic.  As a counterexample
consider the category of sets with the tensor product being the union operation. Then
$(\mathbf{Set}, \cup)$ is a monoidal category; yet if $A$ and $B$ are
non-disjoint sets, then \begin{tikzcd}A \arrow[>->]{r}{\subseteq} & A\cup B \arrow[<-<]{r}{\subseteq} & B \end{tikzcd} are not $\Mor$-independent
(see subsection \ref{sec:examples-sets}). Also note that there are tensorial categories where components of a
tensor product cannot even be mapped into the tensor product hence they are not subobjects (an example is the category of rings with homomorphisms). This motivates Definition \ref{def:regular-category} below. Note that $\oprod$ being a bifunctor means that it acts on $\Mor\times\Mor$ too; that is to say: if $f:A\to A'$ and $g:B\to B'$ are two morphisms, then there is a morphism $f\oprod g: A\oprod B\to A'\oprod B'$.

\begin{definition}\label{def:regular-category}
	 The tensorial category $(\Ob, \Mor, \oprod)$ is called
	{\em $\Hom$-regular} if (i) and (ii) below hold.
	\begin{itemize}
		\item[(i)] For all objects $A, B$ there are $\Hom$-monomorphisms
		 \begin{center}\begin{tikzcd} A \arrow[>->]{r}{i_A} & A\oprod B \arrow[<-<]{r}{i_B} & B \end{tikzcd}\end{center}
		We call these $\Hom$-monomorphisms {\em canonical injections}.
			
		\item[(ii)] For any pairs of $\Mor$-morphisms $m_A:A\to A'$ and $m_B:B\to B'$
		the tensor product arrow $m_A\oprod m_B$ makes the following
		diagram commute.
		\begin{center}
				\begin{tikzpicture}
				  \matrix (m) [matrix of math nodes,row sep=3em,column sep=4em,minimum width=2em]
				  {
					 A  & A \oprod B  & B\\
					 A'  & A' \oprod B'  & B'\\};
				  \path[-stealth]
					(m-1-1) edge[->] node [above] {$i_{A}$} (m-1-2)
							edge node [left]  {$m_A$} (m-2-1)
					(m-2-1) edge[->] node [below] {$i_{A'}$} (m-2-2)
					(m-1-3) edge[->] node [above] {$i_{B}$} (m-1-2)
							edge node [right] {$m_B$} (m-2-3)
					(m-2-3) edge[->] node [below] {$i_{B'}$} (m-2-2)
					(m-1-2) edge[dashed] node [left] {$m_A\oprod m_B$} (m-2-2);
				\end{tikzpicture}	
		\end{center}			
	\end{itemize}\endef
\end{definition}
We then have as an immediate consequence of regularity:
\begin{corollary}\label{cor:tensor-sufficient}
	If $(\Ob, \Mor, \oprod)$ is a {\em $\Hom$-regular} tensorial
	category, then
	the canonical injections are $\Mor$-independent.
\end{corollary}
\begin{proof}
	Take $A'=A$ and $B'=B$ in the definition of regularity.
\end{proof}

\begin{definition}[$(\Hom, \oprod)$-independence of $\Hom$-morphisms]
	$\Hom$-morphisms \begin{tikzcd}A \arrow[->]{r}{f_A} & C \arrow[<-]{r}{f_B} & B \end{tikzcd}
	are called $(\Hom, \oprod)$-independent if there exists a
	$\Hom$-morphism	$h:A\otimes B\to C$ such that the following diagram commutes.
	\begin{center}
		\begin{tikzpicture}
		  \matrix (m) [matrix of math nodes,row sep=3em,column sep=4em,minimum width=2em]
		  {
		     A & A\otimes B & B  \\
		       & C &  \\};
		  \path[-stealth]
		    (m-1-1) edge node [left] {$f_A$} (m-2-2)
		    (m-1-3) edge node [right] {$f_B$} (m-2-2)
		    (m-1-1) edge[>->] node [above] {$i_A$} (m-1-2)
		    (m-1-3) edge[>->] node [above] {$i_B$} (m-1-2)
		    (m-1-2) edge [dashed] node [left] {$h$} (m-2-2);
		\end{tikzpicture}	
	\end{center} \endef
\end{definition}

\begin{definition}[$\Hom$-injectivity]
	An object $Q$ is $\Hom$-injective if for all $A$ and arrows
	in the diagram below
	we have
	\begin{center}\begin{tikzpicture}
				  \matrix (m) [matrix of math nodes,row sep=3em,column sep=4em,minimum width=2em]
				  {
					 A  & Q  \\
					    & Q  \\};
				  \path[-stealth]
					(m-1-1) edge[->] node [above] {$\in\Hom$} (m-1-2)
							edge node [below] {$\in\Mor$} (m-2-2)
					(m-1-2) edge [dotted] node [right] {$\exists\in\Mor$} (m-2-2);
	\end{tikzpicture}\end{center}
	\endef
\end{definition}

\begin{proposition}\label{prop:TensorMorIndep}
	In a $\Hom$-regular tensorial category $(\Ob, \Mor, \oprod)$
	we have that
	$(\Hom,\oprod)$-independence of $\Hom$-subobjects in a $\Hom$-injective object
	implies $\Mor$-independence of the $\Hom$-subobjects.
\end{proposition}
\begin{proof}
	Suppose
	\begin{tikzcd}A \arrow[>->]{r}{f_A} & Q \arrow[<-<]{r}{f_B} & B \end{tikzcd}
	are $\Hom$-monomorphisms representing two $\Hom$-subobjects which are
	$(\Hom,\oprod)$-independent. Let $Q$ be $\Hom$-injective and
	consider the diagram below.
	
	
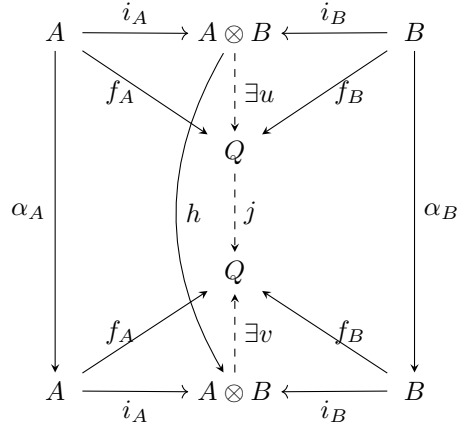
\begin{center}
		\begin{tikzpicture}
		  \matrix (m) [matrix of math nodes,row sep=3em,column sep=4em,minimum width=2em]
		  {
		     A & A\otimes B & B  \\
		       & Q &  \\
			   & Q & \\
		     A & A\otimes B & B  \\};
		  \path[-stealth]
		    (m-1-1) edge[->] node [above] {$i_A$} (m-1-2)
		    (m-1-3) edge[->] node [above] {$i_B$} (m-1-2)
		    (m-1-1) edge node [left] {$f_A$} (m-2-2)
		    (m-1-3) edge node [right] {$f_B$} (m-2-2)
		    (m-1-2) edge [dashed] node [right] {$\exists u$} (m-2-2)
			(m-1-1) edge node [left] {$\alpha_A$} (m-4-1)
			(m-1-3) edge node [right] {$\alpha_B$} (m-4-3)
		    (m-4-1) edge[->] node [below] {$i_A$} (m-4-2)
		    (m-4-3) edge[->] node [below] {$i_B$} (m-4-2)
		    (m-4-1) edge node [left] {$f_A$} (m-3-2)
		    (m-4-3) edge node [right] {$f_B$} (m-3-2)
		    (m-4-2) edge [dashed] node [right] {$\exists v$} (m-3-2)
			(m-1-2) edge [bend right] node [right] {$h$} (m-4-2)
			(m-2-2) edge [dashed] node [right] {$j$} (m-3-2);
		\end{tikzpicture}	
	\end{center}
	
	By $(\Hom,\oprod)$-independence of
	\begin{tikzcd}A \arrow[>->]{r}{f_A} & Q \arrow[<-<]{r}{f_B} & B \end{tikzcd}
	there is $u, v:A\oprod B\to Q$ with $u,v\in\Hom$ and by
	regularity of the tensorial
	category there is $h:A\oprod B\to A\oprod B$, $h = \alpha_A\oprod\alpha_B$
	making the diagram commute.
	Applying $\Hom$-injectivity of $Q$ for $u$ and $hv$ we get a suitable
	$j:Q\to Q$ with $uj=hv$. Then
	\[  f_Aj = i_Auj = i_Ahv = \alpha_Ai_Av = \alpha_Af_A \]
	and similarly
	\[  f_Bj = i_Buj = i_Bhv = \alpha_Bi_Bv = \alpha_Bf_B \]	
\end{proof}

The intuitive content of Proposition \ref{prop:TensorMorIndep} is as follows.
Suppose $A$ and $B$ are $\Hom$-subobjects of an $\Hom$-injective
object $Q$. The subobject relations are witnessed by the $\Hom$-arrows
$f_A$ and $f_B$. $(\Hom,\oprod)$-independence tells us that $A$ and $B$,
as subobjects, lie in $Q$ in a similar manner as they lie in the tensor product $A\oprod B$, i.e. the tensor product can be mapped into $Q$
via some $\Hom$-arrow $u$ in such a way that the canonical injections (that witness that $A$ and $B$ are $\Hom$-subobjects of the tensor product) commute with $f_A$, $f_B$ and $u$. Take any two
$\Mor$-morphisms $\alpha_A: A\to A$ and $\alpha_B: B\to B$. By $\Hom$-regularity of the tensor product this two mappings are
{\em jointly} implementable by a single morphism $h$ {\em on the tensor product}. The question is whether this mapping $h$ can be extended to a mapping defined on the entire $Q$. $\Hom$-injectivity of $Q$ does this
favour to us: $\Hom$-injectivity guarantees that any $\Mor$-morphism
defined on a $\Hom$-subobject can be extended as a $\Mor$-morphism
acting on $Q$.

$\oprod$-independence of $\Mor$-morphisms (Definition \ref{def:tensindep}) and the notion of a regular category (Definition \ref{def:regular-category}) was introduced and studied in \cite{Franz2002} under different names. In \cite{Franz2002} the notion of a {\em tensor product with projections} or {\em with inclusions} has been defined (essentially, this is our Definition \ref{def:regular-category}). It was shown in \cite{Franz2002} that the definition of stochastic independence relies on such a structure and that independence can be defined in an arbitrary category with a tensor product with inclusions or projections in a manner similar to Definition \ref{def:tensindep}. It turns out that the standard notion of stochastic independence of classical random variables is equivalent to $\oprod$-independence of objects in the category of random variables (for more detail see \cite{Franz2002}). Moreover, the classifications of quantum stochastic independence by Muraki, Ben Ghorbal, and Sch\"urmann has been shown to be classifications of the tensor products with inclusions for the categories of algebraic probability spaces and non-unital algebraic probability spaces. Thus $\oprod$-independence of $\Mor$-morphisms is directly relevant for stochastic independence in the context of quantum probability spaces.

\section{Subsystem independence as subobject independence in the category of \C algebras with respect to operations as morphisms\label{sec:opind}}

In this section $(\mathfrak{Alg}, Op_{\mfAlg})$ denotes the category of \C algebras, where the elements in the class of morphisms $Op_{\mfAlg}$ are the non-selective operations: completely positive, unit preserving linear maps on \C algebras. Operations represent physical operations performed on quantum physical systems whose algebra of observables are represented by the (selfadjoint) part of the \C algebra the operation is defined on. Examples of operations include states, conditional expectations (in particular the projection postulate), operations that are given by Kraus operators, and more (see \cite{Kraus1983} for the elementary theory and physical interpretation of operations, \cite{Arveson1969} for some basic properties, and \cite{Paulsen2003} for a systematic treatment of operations from the perspective of operator spaces.) Specifically, \C algebra homomorphisms are completely positive; hence the class $hom_{\mathfrak{Alg}}$ of injective \C algebra homomorphism is a subclass of $Op_{\mfAlg}$. Thus it is meaningful to talk about $Op_{\mfAlg}$-independence in the sense of the following definition:

\begin{definition}\label{def:opind}
{\rm
\C subalgebras $\cA,\cB$ of \C algebra $\cC$ are called $Op_{\mfAlg}$-independent in $\cC$ if $\cA$ and $\cB$ are $Op_{\mfAlg}$-independent as $hom_{\mathfrak{Alg}}$-subobjects of object $\cC$ in the category $(\mathfrak{Alg}, Op_{\mfAlg})$ of \C algebras in the sense of Definition \ref{def:subobj-indep}.
}
\endef
\end{definition}

The notion of $Op_{\mfAlg}$-independence of \C subalgebras was first formulated in categorial terms in \cite{Redei2014SHPMP} but its content, expressed in a non-categorial terminology and called ``operational \C independence" appeared already in \cite{Redei-Summers2009}. The content of $Op_{\mfAlg}$-independence of \C subalgebras $\cA,\cB$ of \C algebra $\cC$ is that operations on the \C subalgebras $\cA,\cB$ have a joint extension to the \C algebra $\cC$. This kind of independence has a direct physical interpretation: The physical content of $Op_{\mfAlg}$-independence is that any two physical operations (for instance measurement interaction) performed on the two subsystems observables of which are represented by $\cA$ and $\cB$, respectively, can be performed as a single physical operation on the larger system observables of which are represented by $\cC$.

Note that $Op_{\mfAlg}$-independence of $\cA, \cB$ in $\cC$ has two components: (i) that operations on $\cA$ and $\cB$ \emph{can} be extended to $\cC$; and (ii) that there exists a \emph{joint} extension. Already (i) is a non-trivial demand because operations on \C subalgebras are not always extendable to the larger algebra \cite{Arveson1969}. Formulated differently: Not all \C algebras are injective. This fact complicates the implementation of subsystem independence as $Op_{\mfAlg}$-independence in the categorial formulation of quantum field theory (see the end of the final section of the paper). Also note that $Op_{\mfAlg}$-independence does not require that the extension of the operations on $\cA$ and $\cB$ factorize across $\cA$ and $\cB$; i.e. the extension need not be a product extension. One can strengthen the notion of $Op_{\mfAlg}$-independence by requiring the existence of a product extension; we call the resulting concept of independence $Op_{\mfAlg}$-independence \emph{in the product sense}.

$(\mathfrak{Alg}, Op_{\mfAlg})$ is a tensor category with respect to the minimal \C tensor product $\cA\otimes\cB$ of \C algebras $\cA$ and $\cB$. Since algebras $\cA$ and $\cB$ have units $I_{\cA}, I_{\cB}$, they can be injected into the tensor product by the $hom_{\mathfrak{Alg}}$-morphisms $\cA\ni A\mapsto A\otimes I_{\cB}$ and $\cB\ni B\mapsto I_{\cA}\otimes B$. Thus the canonical $hom_{\mathfrak{Alg}}$-injections in
Definition \ref{def:regular-category} (i) exist and item (ii) in Definition
\ref{def:regular-category} is also fulfilled. Thus $(\mathfrak{Alg}, Op_{\mfAlg}, \oprod)$ is a $hom_{\mathfrak{Alg}}$-regular category in the sense of Definition \ref{def:regular-category}. It follows that Proposition \ref{prop:TensorMorIndep} applies and we obtain

\begin{proposition}\label{prop:C*tensor-indep}
\C algebras $\cA\thickapprox\cA\otimes I_{\cB}$ and $\cB\thickapprox I_{\cA}\otimes \cB$ are $Op_{\mfAlg}$-independent in $\cA\otimes\cB$.
\end{proposition}
As a corollary:
\begin{corollary}\label{cor:indep-tensor-inject}
If $\cC$ is an injective \C algebra and $\cA\otimes\cB$ is a \C subalgebra of $\cC$, then $\cA\thickapprox\cA\otimes I_{\cB}$ and $\cB\thickapprox I_{\cA}\otimes \cB$ are $Op_{\mfAlg}$-independent in $\cC$.
\end{corollary}
The joint extension to $\cA\otimes\cB$ of operations on $\cA$ and $\cB$ guaranteed by Proposition \ref{prop:C*tensor-indep} is just the tensor product of the two operations, which is again an operation \cite{Blackadar2005}[p. 190], (see also Proposition 9. in \cite{Redei-Summers2009}).
Note that \C algebras $\cA$, $\cB$ are not just $Op_{\mfAlg}$-independent in $\cA\otimes\cB$, they are $Op_{\mfAlg}$-independent in $\cA\otimes\cB$ \emph{in the product sense}: the tensor product of two operations factorizes over the components. $Op_{\mfAlg}$-independence in the product sense is a very strong independence property. It is known to be strictly stronger than $Op_{\mfAlg}$-independence simpliciter: $Op_{\mfAlg}$-independence in $\cC$ of commuting \C subalgebras $\cA$, $\cB$ of $\cC$ in the product sense is equivalent to \C independence of $\cA,\cB$ in the product sense (Proposition 10, \cite{Redei-Summers2009}) but \C independence of $\cA,\cB$ is strictly weaker than \C independence of $\cA,\cB$ in the product sense \cite{Summers1990b} (cf. Proposition  1. in \cite{Redei-Summers2009}).

The difference between $Op_{\mfAlg}$-independence and $Op_{\mfAlg}$-independence in the product sense, and the fact that the latter concept relies on the morphisms in $Op_{\mfAlg}$ being functions, lead to the question of whether there is a purely categorial version of subobject independence as morphism co-possibility ``in the product sense". We do not have such a concept and leave it is a problem for further investigation.

Note that the definition of $\Mor$-independence of $\Hom$-subobjects (Definition \ref{def:subobj-indep}) remains meaningful even if the class $\Hom$ is not a subclass of $\Mor$: As long as morphisms in $\Hom$ and $\Mor$ can be composed, one can meaningfully talk about $\Mor$-independence of $\Hom$-subobjects. This enables one to recover the major subsystem independence concepts that occur in algebraic quantum (field) theory by choosing special subclasses of the class of all non-selective operations $Op_{\mfAlg}$. For instance, taking states as a subclass of operations, one obtains \C independence; if algebras $\cA,\cB$ and $\cC$ are von Neumann algebras, taking \emph{normal} states as the subset of operations one obtains \W independence; taking \emph{normal} operations as subclass of operations, one obtains operational \W independence  (cf. \cite{Redei2010FoundPhys}). One also can define the product versions of these specific independence concepts by considering $Op_{\mfAlg}$-independence in the product sense with respect to the respective subclasses of operations. One has then the notions of \C and \W independence in the product sense, and operational \C and \W independence in the product sense. Specifications of further sub-types of independence obtains by considering particular operations such as conditional expectations, or Kraus operations (see \cite{Redei2010FoundPhys}). The logical relation of these independence concepts emerges then as a non-trivial problem, some of which are still open \cite{Redei2010FoundPhys}.  Viewed from the perspective of the resulting hierarchy of independence notions, $Op_{\mfAlg}$-independence serves as a general, categorial frame in which independence can be formulated and analyzed.

Given the concept of $Op_{\mfAlg}$-independence, it is natural to consider it as a possible condition to impose it on the covariant functor $\cF$ representing quantum field theory in order to express causal locality in terms of it. To do so we recall first the definition of the functor $\cF$ describing quantum field theory.

\section{$Op_{\mfAlg}$-independence as locality condition in categorial quantum field theory \label{sec:functor-of-QFT}}
The functor $\cF$ representing a general covariant quantum field theory is between two categories:
(i) $(\mathfrak{Man}, {hom}_{\mathfrak{Man}})$, the category of spacetimes with isometric embeddings of spacetimes as morphisms; and (ii)
$(\mathfrak{Alg}, hom_{\mathfrak{Alg}})$, the category of \C algebras  with injective \C algebra homomorphisms as morphisms.
The category  $(\mathfrak{Man}, {hom}_{\mathfrak{Man}})$ is specified by the following stipulations (see \cite{Brunetti-Fredenhagen-Verch2003} for more details):
\begin{itemize}\itemsep-1pt
\item[(i)] The objects in $Obj(\mathfrak{Man})$ are 4 dimensional $C^{\infty}$ spacetimes $(M,g)$ with a Lorentzian metric $g$ and such that
$(M,g)$ is Hausdorff, connected, time oriented and globally hyperbolic.
\item[(ii)] The morphisms in $hom_{\mfMan}$ are isometric smooth embeddings $\psi\colon (M_1,g_1)\to (M_2,g_2)$
that preserve the time orientation and are causal in the following sense:  if the endpoints $\gamma(a),\gamma(b)$ of a timelike curve
$\gamma \colon [a,b]\to M_2$ are in the image $\psi (M_1)$, then the whole curve is in the image:
$\gamma(t)\in\psi (M_1) \mbox{\ for all\ }  t\in [a,b]$.
The composition of morphisms is the usual composition of maps.
\end{itemize}

\begin{definition}\label{def:QFT-as-functor}
{\rm
A locally covariant quantum field theory is a functor $\cF$ between the categories $(\mathfrak{Man}$, ${hom}_{\mathfrak{Man}})$ and
$(\mathfrak{Alg}, hom_{\mathfrak{Alg}})$: For any object $(M,g)$ in $\mathfrak{Man}$ the $\cF(M,g)$ is a \C algebra in
$\mathfrak{Alg}$; for any homomorphism $\psi$ in ${hom}_{\mathfrak{Man}}$ the $\cF(\psi)$ is an injective \C algebra homomorphism in $hom_{\mathfrak{Alg}}$. The functor $\cF$ is required to have the properties 1.-4. below:
\begin{enumerate}
\item \textbf{Covariance}:
\begin{eqnarray*}
\cF(\psi_1\circ\psi_2)&=&\cF(\psi_1)\circ\cF(\psi_2)\\
\cF(id_{\mfMan})&=&id_{\mfAlg}
\end{eqnarray*}

\item \textbf{Einstein Causality}: Whenever the embeddings $\psi_1: (M_1,g_1)\to (M,g)$ and $\psi_2: (M_2,g_2)\to (M,g)$ are such that $\psi_1(M_1)$ and $\psi_2(M_2)$ are spacelike in $M$, then
\begin{equation}\label{commutator}
\Big\lbrack \cF(\psi_1)\Big(\cF(M_1,g_1)\Big),\cF(\psi_2)\Big(\cF(M_2,g_2)\Big)\Big\rbrack^{\cF(M,g)}_-=\{0\}
\end{equation}
where $\lbrack \ , \ \rbrack^{\cF(M,g)}_-$ in (\ref{commutator}) denotes the commutator in the \C algebra $\cF(M,g)$.

\item \textbf{Time slice axiom}:
If $(M,g)$ and $(M', g')$ and the embedding $\psi\colon (M,g)\to (M',g')$
are such that $\psi(M,g)$ contains a Cauchy surface for $(M',g')$ then
$$
\cF(\psi)\cF(M,g)=\cF(M',g')
$$

\item $Op_{\mfAlg}$-\textbf{independence}:
Whenever the embeddings
$\psi_1: (M_1,g_1)\to (M,g)$
and $\psi_2:  (M_2,g_2)\to (M,g)$ are such that $\psi_1(M_1)$ and $\psi_2(M_2)$ are spacelike in $M$, then the objects $\cF(M_1,g_1)$ and $\cF(M_2,g_2)$ are $Op_{\mfAlg}$-independent in $\cF(M,g)$ in the sense of Definition \ref{def:opind}.
\end{enumerate}
}
\endef
\end{definition}
The axiom system specified by Definition \ref{def:QFT-as-functor} differs from the one originally proposed in \cite{Brunetti-Fredenhagen-Verch2003} by the addition of the $Op_{\mfAlg}$-independence condition. Following the terminology introduced in \cite{Redei2016CatLocNagoya}, we call the original axiom system in \cite{Brunetti-Fredenhagen-Verch2003} \textbf{BASIC}, to distinguish it from the one given by Definition \ref{def:QFT-as-functor}, which we call \textbf{OPIND}. One also can strengthen \textbf{OPIND} by requiring in 4. in Definition \ref{def:QFT-as-functor} that the objects $\cF(M_1,g_1)$ and $\cF(M_2,g_2)$ are $Op_{\mfAlg}$-independent in $\cF(M,g)$ in the \emph{product sense}. We call the resulting axiom system $\textbf{OPIND}^{\times}$.

Other stipulations on the functor are also possible and have been formulated: The axiom system \textbf{BASIC} was amended by Brunetti and Fredenhagen by replacing the Einstein Causality condition by an axiom that requires a tensorial property of $\cF$ (Axiom 4 in \cite{Brunetti-Fredenhagen2009}; also see \cite{Fredenhagen-Rejzner2012}). To formulate this axiom one first extends $(\mathfrak{Man}, {hom}_{\mathfrak{Man}})$ to a tensor category $(\mathfrak{Man}^{\otimes}, {hom}^{\otimes}_{\mathfrak{Man}})$. This tensor category has, by definition, as its objects \emph{finite} disjoint unions of objects from $\mathfrak{Man}$, and the empty set as unit object. The morphisms $h^{\otimes}$ in ${hom}^{\otimes}_{\mathfrak{Man}}$ are embeddings of unions of disjoint spacetimes that are ${hom}_{\mathfrak{Man}}$-homomorphisms when restricted to the (connected) elements of the disjoint union of spacetimes and have the feature that the images under $h^{\otimes}$ of disjoint spacetimes are spacelike. The functor $\cF$ is then required to be extendable to a tensor functor $\cF^{\otimes}$ between $(\mathfrak{Man}^{\otimes}, {hom}^{\otimes}_{\mathfrak{Man}})$ in a natural way.  We call the resulting axiom system  \textbf{TENSOR}.

One obtains yet another axiom system if one requires a categorial version of the split property. This condition was formulated in \cite{Brunetti-Fredenhagen-Paniz-Reijzner2014} -- together with the categorial version of weak additivity. The definitions are:
\begin{definition}
{\rm
The functor $\cF$ has the categorial split property if the following two conditions hold:

\begin{enumerate}
\item For spacetimes $(M,g_M), (N,g_N)$ in $\mathfrak{Man}$ and morphism $\psi\colon (M,g_M) \to (N ,g_N)$ such that the closure of $\psi(M,g_M)$ is compact, connected and in the interior of $M$, there exists a type $I$ von Neumann factor $\cR$ such that
\begin{equation}
\cF(\psi)(\cF(M,g_M))\subset \cR\subset \cF(N,g_N)
\end{equation}
\item $\sigma$-continuity of the $\cF(\psi')$ with respect to the inclusion $\cR\subset \cR'$, where $\psi': (M,g_M)\to (L,g_L)$    and
\begin{eqnarray}
(\cF(\psi')\circ\cF(\psi))(\cF(M,g_M))&\subset& \cF(\psi')(\cR)\\
\subset \cF(\psi')(\cF(N,g_N))&\subset& \cR'\subset \cF(L,g_L)
\end{eqnarray}
\end{enumerate}
}
\endef
\end{definition}

\begin{definition}[weak additivity of the functor $\cF$]
{\rm
The functor $\cF$ satisfies weak additivity if for any spacetime $(M,g)$ and any family of spacetimes $(M_i,g_i)$ with morphisms $\psi_i\colon (M_i,g_i)\to (M,g)$ such that
\begin{equation}
M\subseteq \cup_i \psi_i(M_i)
\end{equation} we have
\begin{equation}
\cF(M,g)=\overline{\cup_i \cF(\psi_i)(\cF(M_i,g_i)))}^{norm}
\end{equation}
}
\endef
\end{definition}
We call \textbf{BASIC+SPLIT} the axiom system that requires of the covariant functor $\cF$ to have weak additivity and the categorial split property, in addition to Einstein Locality and Time Slice axiom.

As these different conditions imposed on the functor $\cF$ show, one can articulate the concept of physical locality understood as independence of the algebras of observables of spatio-temporaly local physical systems localized in  causally disjoint spacetime regions in more than one way. Thus the question or relation of the different axiom systems arise, and one also can ask: which one of the axiom systems is the most adequate.

The problem of the relation of the axiom  systems was raised in \cite{Redei2016CatLocNagoya}, where it was argued that the implications in the following diagram depicting the logical relations hold. Here we comment on the reverse of the indicated implications below.

\[\begin{array}{ccccccc}
	\textbf{TENSOR} &  \Leftarrow & \textbf{OPIND}^{\times} & \Rightarrow & \textbf{OPIND} & \Rightarrow &
	\textbf{BASIC} \\
	     \Updownarrow    & &                         & &                & & \\
	\textbf{BASIC + SPLIT} & &  & & & & \\
\end{array}\]


We have seen in Section \ref{sec:motiv} that \textbf{BASIC} does not entail \textbf{OPIND}.
The technical obstacle prohibiting the reverse of the implication $\textbf{OPIND}{}^{\times}$ $\Rightarrow$ \textbf{TENSOR} to hold trivially is that operations on \C subalgebras of a \C algebra $C$ need not be extendable to $C$. Hence, although \C subalgebras $\cA,\cB$ are $Op_{\mfAlg}$-independent in the tensor product $\cA\otimes\cB$, this does not entail without further conditions that $\cA,\cB$ are $Op_{\mfAlg}$-independent in a \C algebra $\cC$ containing $\cA\otimes\cB$ as a \C subalgebra. Injectivity of $\cC$ would entail this; however, it is not clear to us whether the  \C algebras $\cF(M,g)$ are injective in general -- or at least for some specific, typical spacetime regions such as double cones.

The reverse of the implication $\textbf{OPIND}{}^{\times}$  $\Rightarrow$ \textbf{OPIND} is unlikely to hold, given that operational \C independence in the product sense is a strictly stronger independence condition than operational \C independence -- but we do not have a rigorous proof of $\textbf{OPIND}$  $\not\Rightarrow$ $\textbf{OPIND}{}^{\times}$ in terms of a model of the axioms displaying the non-implication.

In view of the logical (in)dependencies of the axiom systems depicted in the chart, the conclusion we propose is that the most natural independence condition to stipulate to hold for the functor $\cF$ in order to express physical locality is $Op_{\mfAlg}$-independence. This condition has a very natural physical interpretation and it does not require more than what is contained in the notion of subsystem independence as co-possibility. So, if some physically relevant models existed which violate \textbf{TENSOR} but satisfy $\textbf{OPIND}$, that model would still be entirely acceptable from the perspective of a causal behavior of the quantum filed theory represented by the functor satisfying $\textbf{OPIND}$.

Our final remark concerns a possible characterization of spacelike separatedness of spacetime regions as subobject independence with respect to some embeddings of spacetimes as morphisms.\footnote{This problem was raised by K. Fredenhagen in the discussion after a talk based on this paper was delivered at the ``Local Quantum Physics and Beyond -- in Memoriam Rudolf Haag'', September 26-27, Hamburg, Germany.} Specifically, one would like to know if the causal embeddings defining the homomorphisms in the category $(\mathfrak{Man}, {hom}_{\mathfrak{Man}})$ have this feature. If indeed ${hom}_{\mathfrak{Man}}$-independence of spacetimes in the category $(\mathfrak{Man}, {hom}_{\mathfrak{Man}})$ (in the sense of Definition \ref{def:subobj-indep}) entails spacelike separetedness, then causal locality of the functor $\cF$ could be defined in a nice, compact manner as independence-faithfulness of the functor, where independence both in the domain and in the range of $\cF$ is captured completely by categorial subobject independence with respect to natural classes of morphisms.

\section*{Acknowledgement}
Research supported in part by the Hungarian Scientific Research Found (OTKA). Contract numbers: K 115593 and K 100715. Zal\'an Gyenis was partially supported by the Premium Postdoctoral Grant of the Hungarian Academy of Sciences.
\small
\bibliography{RedeiBib}

\begin{thebibliography}{10}

\bibitem{Araki1999}
H.~Araki.
\newblock {\em Mathematical Theory of Quantum Fields}, volume 101 of {\em
  International Series of Monograps in Physics}.
\newblock Oxford University Press, Oxford, 1999.
\newblock Originally published in Japanese by Iwanami Shoten Publishers, Tokyo,
  1993.

\bibitem{Arveson1969}
W.~Arveson.
\newblock Subalgebras of ${C}^*$-algebras.
\newblock {\em Acta Mathematica}, 123:141--224, 1969.

\bibitem{Awodey2010}
S.~Awodey.
\newblock {\em Category Theory}.
\newblock Oxford University Press, 2010.
\newblock Second edition.

\bibitem{Blackadar2005}
B.~Blackadar.
\newblock {\em Operator Algebras: {T}heory of {C}*-Algebras and von {N}eumann
  Algebras}.
\newblock Encyclopaedia of Mathematical Sciences. Springer, 1. edition, 2005.

\bibitem{Brunetti-Fredenhagen2006}
R.~Brunetti and K.~Fredenhagen.
\newblock Algebraic approach to quantum field theory.
\newblock In Jean-Pierre Francoise, Gregory~L. Naber, and Tsou~Sheung Tsun,
  editors, {\em Elsevier Encyclopedia of Mathematical Physics}, pages 198--204.
  Academic Press, Amsterdam, 2006.
\newblock arXiv:math-ph/0411072.

\bibitem{Brunetti-Fredenhagen2009}
R.~Brunetti and K.~Fredenhagen.
\newblock Quantum field theory on curved backgrounds.
\newblock In C.~B\"ar and K.~Fredenhagen, editors, {\em Quantum Field Theory on
  Curved Spacetimes}, volume 786 of {\em Lecture Notes in Physics}, chapter~5,
  pages 129--155. Springer, Dordrecht, Heidelberg, London, New York, 2009.

\bibitem{Brunetti-Fredenhagen-Paniz-Reijzner2014}
R.~Brunetti, K.~Fredenhagen, I.~Paniz, and K.~Rejzner.
\newblock The locality axiom in quantum field theory and tensor products of
  ${C}^*$-algebras.
\newblock {\em Reviews of Mathematical Physics}, 26:1450010, 2014.
\newblock arXiv:1206.5484 [math-ph].

\bibitem{Brunetti-Fredenhagen-Verch2003}
R.~Brunetti, K.~Fredenhagen, and R.~Verch.
\newblock The generally covariant locality principle -- a new paradigm for
  local quantum field theory.
\newblock {\em Communications in Mathematical Physics}, 237:31--68, 2003.
\newblock arXiv:math-ph/0112041.

\bibitem{Buchholz2001}
D.~Buchholz.
\newblock Algebraic quantum field theory: {A} status report.
\newblock In A.~Grigoryan, A.~Fokas, T.~Kibble, and B.~Zegarlinski, editors,
  {\em XIIIth International Congress on Mathematical Physics, Imperial College,
  London, UK}, pages 31--46. International Press of Boston, Sommervile, MA
  U.S.A., 2001.
\newblock arXiv:math-ph/0011044.

\bibitem{Buchholz-Haag2000}
D.~Buchholz and R.~Haag.
\newblock The quest for understanding in relativistic quantum physics.
\newblock {\em Journal of Mathematical Physics}, 41:3674--3697, 2000.
\newblock arXiv:hep-th/9910243.

\bibitem{Buchholz-Summers2005}
D.~Buchholz and S.J. Summers.
\newblock Quantum statistics and locality.
\newblock {\em Physics Letters A}, 337:17--21, 2005.

\bibitem{Franz2002}
U.~Franz.
\newblock What is stochastic independence?
\newblock In N.~Obata, T.~Matsui, A.~Hora, and S{\=u}ri Kaiseki~Kenky{\=u}jo
  Ky{\=o}to~Daigaku, editors, {\em Non-commutativity, Infinite-dimensionality
  and Probability at the Crossroads: Proceedings of the RIMS Workshop on
  Infinite-Dimensional Analysis and Quantum Probability: Kyoto, Japan, 20-22
  November, 2001}, QP-PQ Quantum Probability and White Noise Analysis, pages
  254--274. World Scientific, 2002.
\newblock arxiv.org/abs/math/0206017.

\bibitem{Fredenhagen2010}
K.~Fredenhagen.
\newblock Lille 1957: The birth of the concept of local algebras of
  observables.
\newblock {\em The European Physical Journal H}, 35:239--241, 2010.

\bibitem{Fredenhagen-Reijzner2016}
K.~Fredenhagen and K.~Reijzner.
\newblock Quantum field theory on curved spacetimes: Axiomatic framework and
  examples.
\newblock {\em Journal of Mathematical Physics}, 57:031101, 2016.

\bibitem{Fredenhagen-Rejzner2012}
K.~Fredenhagen and K.~Rejzner.
\newblock Local covariance and background independence.
\newblock In F.~Finster, O.~M\"uller, M.~Nardmann, J.~Tolksdorf, and
  E.~Zeidler, editors, {\em Quantum Field Theory and Gravity. Conceptual and
  Mathematical Advances in the Search for a Unified Framework}, pages 15--24.
  Birkh\"auser Springer Basel, Basel, 2012.
\newblock arXiv:1102.2376 [math-ph].

\bibitem{Haag1992}
R.~Haag.
\newblock {\em Local Quantum Physics: Fields, Particles, Algebras}.
\newblock Springer Verlag, Berlin and New York, 1992.

\bibitem{Haag2010a}
R.~Haag.
\newblock Discussion of the `axioms' and the asymptotic properties of a local
  field theory with composite particles.
\newblock {\em The European Physical Journal H}, 35:243--253, 2010.
\newblock English translation and re-publication of a talk given at the
  international conference on mathematical problems of the quantum theory of
  fields, Lille, June 1957.

\bibitem{Haag2010b}
R.~Haag.
\newblock Local algebras. {A} look back at the early years and at some
  achievements and missed opportunities.
\newblock {\em The European Physical Journal H}, 35:255--261, 2010.

\bibitem{Haag-Kastler1964}
R.~Haag and D.~Kastler.
\newblock An algebraic approach to quantum field theory.
\newblock {\em Journal of Mathematical Physics}, 5:848--861, 1964.

\bibitem{Horuzhy1990}
S.S. Horuzhy.
\newblock {\em Introduction to Algebraic Quantum Field Theory}.
\newblock Kluwer Academic Publishers, Dordrecht, 1990.

\bibitem{Kalmbach1983}
G.~Kalmbach.
\newblock {\em Orthomodular Lattices}.
\newblock Academic Press, London, 1983.

\bibitem{Kraus1983}
K.~Kraus.
\newblock {\em States, Effects and Operations}, volume 190 of {\em Lecture
  Notes in Physics}.
\newblock Springer, New York, 1983.

\bibitem{Maeda1949}
F.~Maeda.
\newblock Direct sums and normal ideals of lattices.
\newblock {\em Journal of Science of the Hiroshima University. Series A},
  14:85--92, 1949.

\bibitem{Paulsen2003}
V.~Paulsen.
\newblock {\em Completely Bounded Maps and Operator Algebras}, volume~78 of
  {\em Cambridge Studies in Advanced Mathematics}.
\newblock Cambridge University Press, Cambridge, 2003.

\bibitem{Pierce2014}
Richard~S. Pierce.
\newblock {\em Introduction to the Theory of Abstract Algebras}.
\newblock Dover Publications, Mineola, New York, 2014.
\newblock Originally published by Holt, Rinehart and Winston, Inc., New York,
  1968.

\bibitem{Redei1995a}
M.~R{\'e}dei.
\newblock Logical independence in quantum logic.
\newblock {\em Foundations of Physics}, 25:411--422, 1995.

\bibitem{Redei1995b}
M.~R{\'e}dei.
\newblock Logically independent von {N}eumann lattices.
\newblock {\em International Journal of Theoretical Physics}, 34:1711--1718,
  1995.

\bibitem{Redei1998}
M.~R{\'e}dei.
\newblock {\em Quantum Logic in Algebraic Approach}, volume~91 of {\em
  Fundamental Theories of Physics}.
\newblock Kluwer Academic Publisher, 1998.

\bibitem{Redei2010FoundPhys}
M.~R\'edei.
\newblock Operational independence and operational separability in algebraic
  quantum mechanics.
\newblock {\em Foundations of Physics}, 40:1439--1449, 2010.

\bibitem{Redei2014SHPMP}
M.~R\'edei.
\newblock A categorial approach to relativistic locality.
\newblock {\em Studies in History and Philosophy of Modern Physics},
  48:137--146, 2014.

\bibitem{Redei2016CatLocNagoya}
M.~R\'edei.
\newblock Categorial local quantum physics.
\newblock In J.~Butterfield, H.~Halvorson, M.~R\'edei, J.~Kitajima, F.~Buscemi,
  and M.~Ozawa, editors, {\em Reality and Measurement in Algebraic Quantum
  Theory}, Proceedings in Mathematics $\&$ Statistics (PROMS). Springer, 2016.
\newblock under review.

\bibitem{Redei-Summers2009}
M.~R\'edei and S.J. Summers.
\newblock When are quantum systems operationally independent?
\newblock {\em International Journal of Theoretical Physics}, 49:3250--3261,
  2010.

\bibitem{Redei-Valente2010}
M.~R\'edei and G.~Valente.
\newblock How local are local operations in local quantum field theory?
\newblock {\em Studies in History and Philosophy of Modern Physics},
  41:346--353, 2010.

\bibitem{Summers1990b}
S.J. Summers.
\newblock On the independence of local algebras in quantum field theory.
\newblock {\em Reviews in Mathematical Physics}, 2:201--247, 1990.

\bibitem{Summers2009}
S.J. Summers.
\newblock Subsystems and independence in relativistic microphysics.
\newblock {\em Studies in History and Philosophy of Modern Physics},
  40:133--141, 2009.
\newblock arXiv:0812.1517 [quant-ph].

\bibitem{Summers2012perspective}
S.J. Summers.
\newblock A perspective on constructive quantum field theory.
\newblock arXiv:1203.3991 [math-ph], 2012.
\newblock This is an expanded version of an article commissioned for UNESCO's
  Encyclopedia of Life Support Systems (EOLSS).

\bibitem{Weinberg1995}
S.~Weinberg.
\newblock {\em The Quantum Theory of Fields. {V}ol. 1: Foundations}.
\newblock Cambridge University Press, Cambridge, 1995.

\end{thebibliography}

\end{document}